\documentclass[a4paper,11pt]{article}
\pdfoutput=1
\usepackage[margin=1in]{geometry}
\usepackage{amsmath}
\usepackage{physics}
\usepackage{dsfont}
\usepackage{xcolor}
\usepackage{mathrsfs}
\usepackage{graphicx}
\usepackage{authblk}
\usepackage[linesnumbered,ruled,vlined]{algorithm2e}

\usepackage[utf8]{inputenc}
\usepackage[backend=biber,style=numeric,url=false,maxbibnames=99]{biblatex}
\usepackage{hyperref}
\hypersetup{colorlinks=true,citecolor=blue,linkcolor=blue,filecolor=blue,urlcolor=blue,breaklinks=true}

\usepackage{amssymb}
\usepackage{amsthm}
\usepackage{pifont}
\usepackage{multirow}

\newtheorem{lemma}{Lemma}
\newtheorem{theorem}{Theorem}
\newtheorem{remark}{Remark}

\newcommand{\id}{\mathds{1}}
\newcommand{\comment}[1]{}
\newcommand{\supp}[0]{\textup{supp}}

\newtheorem{definition}{Definition}

\newcommand{\CC}{\mathbb{C}}
\newcommand{\cG}{\mathcal{G}}
\newcommand{\eps}{\varepsilon}
\DeclareMathOperator{\poly}{poly}

\title{A Faster Algorithm for the Free Energy in One-Dimensional Quantum Systems}
\author{Samuel O. Scalet\thanks{Department of Applied Mathematics and Theoretical Physics, University of Cambridge, United Kingdom \indent\indent sos25@cam.ac.uk}}
\date{}

\begin{document}
\maketitle
\begin{abstract}
We consider the problem of approximating the free energy density of a translation-invariant, one-dimensional quantum spin system with finite range.
While the complexity of this problem is nontrivial due to its close connection to problems with known hardness results, a classical subpolynomial-time algorithm has recently been proposed [Fawzi et al., 2022].
Combining several algorithmic techniques previously used for related problems, we propose an algorithm outperforming this result asymptotically and give rigorous bounds on its runtime.
Our main techniques are the use of Araki expansionals, known from results on the nonexistence of phase transitions, and a matrix product operator construction.
We also review a related approach using the Quantum Belief Propagation [Kuwahara et al., 2018], which in combination with our findings yields an equivalent result.
\end{abstract}
\section{Introduction}

Making predictions of physical quantities in thermal equilibrium for quantum many-body systems is a central goal of condensed matter physics.
Despite a plethora of succesful algorithms succeeding this task in many settings, the field of Hamiltonian complexity has also shown obstacles to this task that cannot be overcome.
Thereby, algorithms are usually either heuristic or specialized to settings in which no complexity theoretic barrier rules out the existence of efficient algorithms.

The \textit{free energy} $F=E-TS$ is an example of such a quantity and has been addressed in many previous works on both, constructive results designing algorithms and complexity theoretic reductions proving hardness.
The interest in the free energy stems from both its immediate physical interest as the value of a functional which variationally characterizes the Gibbs state and also its proximal use in calculations of observables since their computation can be reduced to the computation of derivatives of the free energy.
Two examples of regimes where provably efficient algorithms exist are lattice systems at high temperature \cite{Mann2021} and one dimensional systems at arbitrary nonzero temperature \cite{kuwahara2018,fawzi2022}.
Remarkably, the computational efficiency is intimately linked to the nonexistence of thermal phase transitions, a connection that has also formally been established \cite{harrow2020}.

In this work, we focus on the one-dimensional setting and further specialize to the case of translation-invariant systems in the thermodynamic limit, i.e., infinite system size, which is arguably the physically most relevant.
Despite the reduction of degrees of freedom in the problem input, the existence of efficient algorithms is nontrivial.
The hardness of a related problem, the translation-invariant 1D ground-state problem, which arises as the zero temperature limit of the free energy has been proven in \cite{Gottesman2013,Bausch2017}, see also the discussion in \cite{fawzi2022} on the relation to hardness of the free energy density problem.

Many algorithms in the 1D case are formulated for finite systems.
Since the concept of thermal equilibrium is derived from an infinte system limit, these algorithms will usually be used for large system sizes to approximate the thermodynamic limit.
The error in approximating a system in the thermodynamic limit by a finite one is standardly bounded by the inverse of the system size and consequently this approach does not allow to perform better than linearly in the inverse error.
In the recent work \cite{fawzi2022}, the authors overcome this barrier and introduce a new algorithmic approach that directly treats an infinite system size.
The resulting algorithm runs in subpolynomial time in the inverse error, i.e., in time less than the inverse error to any constant power greater zero.

\paragraph{Setup and main result}
Let us introduce the setup of one-dimensional translation-invariant spin systems in the thermodynamic limit.
The local Hilbert space is $\CC^d$ and the system is defined via a nearest neighbour interaction $h\in\mathcal B(\CC^d\otimes\CC^d)$ with $\|h\|\le1$.
Let us define the Hamiltonian for a finite chain of length $N$ as
\[
H_{[1,N]}=\sum_{i=1}^{N-1}h_{i,i+1},
\]
where $h_{i,i+1}$ has support on sites $i,i+1$.
The case of higher range interactions can be treated by appropriately blocking the system.

The free energy density for a given inverse temperature $\beta$ is defined as
\begin{equation}\label{eq:fdef}
f(h)=-\lim_{N\to\infty} \frac1{\beta N}\log(Z_N)=-\lim_{N\to\infty}\frac1{\beta N}\log(\Tr\left[e^{-\beta H_{[1,N]}}\right]).
\end{equation}

Our main result is the following subpolynomial-time algorithm for approximating the free energy density.
It improves on the asymptotics in the prior work \cite{fawzi2022}, which only gives a runtime of $\exp(\mathcal O(\log(1/\eps)/\log(\log(1/\eps))))$.
\begin{theorem}\label{thm:main}
There exists a classical algorithm taking as input a hermitian matrix $h\in\mathbb{C}^{d^2\times d^2}$ and a precision $1>\eps>0$ and running in time
\[
\exp(\sqrt{ A\log(\frac{ B}\eps)}\log( A\log(\frac{ B}\eps))),
\]
where $A,B$ are constants depending on the inverse temperature and local dimension only, outputting an approximation $\tilde f$ to the one-dimensional translation-invariant free energy density $|f-\tilde f|\le\eps$.
\end{theorem}

\paragraph{Overview of Approach}
We give an overview of our approach to compute the free energy density.
It breaks down into two steps: We first relate the global free energy density to an increment of partition functions, an idea inspired by \cite{kuwahara2018}.
Secondly, we use a matrix product operator (MPO) construction from \cite{kuwahara2021} to compute these partition functions individually.
Combining these error estimates gives our main result.

Compared to \cite{kuwahara2018}, we give a strongly simplified proof of the first step based on Araki expansionals \cite{araki1969,perezgarcia2020}.
\begin{lemma}\label{lem:partFuncRatio}
There exist constants $A, \xi>0$ depending on $\beta$ only such that
\[
\left|e^{-\beta f(h)}-\frac{Z_{l+1}}{Z_l}\right|\le A e^{-l/\xi}
\]
\end{lemma}
A non-quantitative version has been proven in \cite{fawzi2022}.
Our approach is closely related but provides an explicit convergence rate.
While the locality of Araki expansionals holds up to \emph{superexponentially} small errors, our proof involves a result on local indistinguishability that only allows for exponential convergence.
Despite this superexponential decay being essential to obtain a \emph{subpolynomial} runtime in \cite{fawzi2022}, we still achieve this property due to the subexponential runtime in system size and inverse error of the MPO algorithm.

The proof in \cite{kuwahara2018} instead uses a tool called \textit{quantum belief propagation} \cite{hastings2007,capel2023}.
At first sight its range of applicability is much wider allowing for non-translation-invariant and higher dimensional systems and also has a better dependence on the temperature.
However, the proof is conditional on an additional ingredient, a result on the decay of correlations.
While decay of correlations has been proven for one-dimensional translation-invariant infinite \cite{araki1969} and, only recently, finite \cite{bluhm2022} systems, the validity of the condition in \cite{kuwahara2018} is not entirely clear as it is needed for a \emph{perturbed} and finite Gibbs state.
To clarify these connections, we review this approach in Appendix~\ref{sec:QBP} giving a slightly modified proof that replaces the decay of correlations entirely by a local indistinguishability result that has been proven rigorously \cite{perezgarcia2020}.
While not the focus of this manuscript, we note that for higher dimensional systems, the approach based on quantum belief propagation still works and, when further restricting to the high-temperature regime, the decay of correlations has been proven as well, even for non-translation-invariant Gibbs states.
It should be noted that in this regime, the algorithm yields only a quasipolynomial runtime, which falls behind the state of the art cluster expansion techniques \cite{Mann2021}, running in only polynomial time.

We also want to comment on an alternative way of reducing an algorithm for local observables, such as the MPO subroutine we use, to the free energy problem.
The equivalence of free energy and local observable expectation values is folklore knowledge due to a relation known as the derivative trick, giving expectation values of local observables as derivatives of the free energy.
A precise result can be found in \cite{bravyi2022}, which inverts this relation using a discretized integration of the derivative of the free energy.
Suffering a polynomial overhead in the inverse error, this approach does not recover a subpolynomial runtime.
We note, however, the interesting analogy that an iteration over a discretized temperature interval in \cite{bravyi2022} and the iteration over sites in \cite{kuwahara2018}, both yield a product formula for the partition function.

\paragraph{Practical aspects}
Many algorithms with provable runtime bounds, including ours, come with a number of parameters and explicit bounds on those that guarantee the desired accuracy.
These bounds are often impractically large to actually run the provably correct algorithm, but choosing the parameters within the limits of the computational resources available provides an immediate heurestic approach that often gives good approximations.
To this end, it is helpful to limit the number of these parameters, so that this heuristic choice of parameters becomes more feasible.

In our case, the involved quantities are the length scale in Lemma~\ref{lem:partFuncRatio}, as well as two parameters hidden inside the MPO part of the algorithm, a length scale and a degree of polynomial approximation.
It is a major theoretical contribution in \cite{kuwahara2021}, that these parameters yield a polynomial bond dimension \emph{without the need for any truncation}.
Such truncations are common in the literature on tensor network algorithms, but their impact on trace norm errors is uncontrolled\footnote{However, the 2-norm error can be bounded, which renders these truncations useful also theoretically for pure states.}.
We note however, that they have been demonstrated to be succesful in practice \cite{chen2018} and could be integrated into the algorithm proposed here for a heuristic variant.
We note that in \cite{fawzi2022}, an alternative algorithm for the problem we consider, the practical implementation is much simpler since it contains only a single parameter which can be simply chosen as large as allowed by the computational resources and this has been demonstrated to yield good convergence results in practice.
Similarly, the original approach in \cite{kuwahara2018}, calculating the partition functions by diagonalization, would yield a straightforward practical implementation but comes at the cost of not having a theoretical subpolynomial runtime guarantee.

We finally remark, that while our algorithm could be reduced to the local observable problem (see \cite[Appendix B]{fawzi2022}), this can already be achieved using the MPO subroutine on its own directly.
An extension to marginals of larger subsytems has been put forward in \cite{alhambra2021}, which could also be extended to the thermodynamic limit.

\section{Notation}
We consider a multipartite spin system on a chain $\mathbb Z$, where to each site we associate a local Hilbert space $\mathcal{H}_i\cong \mathbb{C}^d$ and the Hilbert space on an interval $[k,l]$ is given by $\mathcal{H}_{[k,l]}=\bigotimes_{i=k}^l\mathcal{H}_i$.
As introduced earlier, we consider a two-site interaction $h\in\mathcal B(\mathbb{C}^d\otimes\mathbb{C}^d)$ and denote by $h_{i,i+1}$ this interaction with support on sites $i,i+1$.
The Hamiltonian for a given interval is then $H_{[k,l]}=\sum_{i=k}^{l-1} h_{i,i+1}$.
We denote by $\Tr[\cdot]$, $\tr[\cdot]=\tr[\cdot]/d^N$ the unnormalized and normalized respectively.
For notational convenience the trace will always be considered to be defined for the space on which the operators act.
We denote the partition function by $Z_N=\Tr[\exp(-\beta H_N)]$ keeping the inverse temperature implicit.
As before the free energy density is given as $f(h)=-1/\beta\lim_{N\to\infty}\log(Z_N)/N$.
We introduce the notation
\[
\rho_{[k,l]}=\frac{e^{-\beta H_{[k,l]}}}{\Tr[e^{-\beta H_{[k,l]}}]}
\]
for the \emph{Gibbs state} on a finite interval .

We introduce a matrix product operator (MPO) on a subsystem $\mathcal{H}_{[k,l]}$ as a set of $(k-l+1)d^2$ $D\times D$ matrices $M_{r,s}^i$ with $1\le r,s\le d$, $k\le i\le l$ defining a state $\rho[M]$ as
\[
\Tr[\rho[M] \bigotimes_{i=k}^l \ket{r_i}\bra{s_i}]=\Tr[\prod_{i=k}^l M_{r_i,s_i}^i].
\]
While MPOs are not used directly in our work, note that the above formula implies directly that the trace of an MPO $\Tr[\rho[M]]$ can be computed in time polynomial in $(k-l,D,d)$.
\section{Reduction to local partition functions}
In this section we prove Lemma~\ref{lem:partFuncRatio}.
We start by introducing some definitions and results regarding Araki expansionals \cite{araki1969}.
The main result of this work is a proof of the analyticity of the free energy, for translation-invariant one-dimensional finite-range systems.
This turns out to be closely related to various results on the decay of correlations as well as the computational complexity of those systems.
The main object of interest is the non-hermitian operator
\[
E_n=e^{-\beta H_{[1,n]}}e^{\beta H_{[2,n]}}.
\]
Due the noncommutativity of the terms in the Hamiltonian, this operator is in general supported on the whole interval $[1,n]$.
The following result shows, however, that it can be well approximated by a local operator and is uniformly bounded in $n$.
Note, that the result holds at arbitrary temperature in the case of finite-range interactions.
The version of the result we cite has been given for systems with exponentially decaying interactions above a threshold temperature, but extend to all temperatures in the finite-range setting.
\begin{lemma}[{\cite[Proposition 4.2]{perezgarcia2020}}]\label{lem:arakiExpansionals}
There is a constant $\cG$ and a superexponentially decaying function $\delta(\cdot)$ only depending on the temperature, but not on $n$ such that for all $a\ge n$
\begin{enumerate}
\item $\|E_n\|, \|E_n^{-1}\|\le\cG$
\item $\|E_n-E_a\|,\|E_n^{-1}-E_a^{-1}\|\le\delta(n)$
\end{enumerate}
\end{lemma}
The same techniques are crucial in proving a \emph{local indistinguishability} result, which we will also require in our analysis.
\begin{lemma}[{\cite[Lemma 4.15]{perezgarcia2020}}]\label{lem:localIndistinguishability}
There exist constants $\mathcal K, \alpha>0$ only depending on $\beta$ such that for any $n'\ge n$ and any observable $A$ supported on $[1,k]$, we have
\[
\left|\Tr[A\rho_{[1,k+n]}]-\Tr[A\rho_{[1,k+n']}]\right|\le \|A\|\mathcal Ke^{-\alpha n}.
\]
\end{lemma}

We now proceed to the proof of Lemma~\ref{lem:partFuncRatio}, a quantitative version of \cite[Lemma 4.1]{fawzi2022} in an analogous fashion.
\begin{proof}[Proof of Lemma~\ref{lem:partFuncRatio}]
We decompose the ratio of partition functions using Araki expansionals
\begin{align*}
\frac{Z_{m+1}}{Z_m}=\frac{\Tr[e^{-\beta H_{[1,m+1]}}]}{\Tr[e^{-\beta H_{[1,m]}}]}=\frac{\Tr[E_{m+1}e^{-\beta H_{[2,m+1]}}]}{\Tr[e^{-\beta H_{[1,m]}}]}=\frac{\Tr[\tr_1[E_{m+1}]e^{-\beta H_{[1,m]}}]}{\Tr[e^{-\beta H_{[1,m]}}]}\\
\end{align*}
Subtracting the same decomposition for $Z_{l+1}/Z_l$ choosing $m\ge l$, we find
\begin{equation}\label{eq:cauchyCrit}
\begin{aligned}
\left|\frac{Z_{m+1}}{Z_m}-\frac{Z_{l+1}}{Z_l}\right|&\le\left|\Tr[\tr_1[E_{l/2}-E_{l+1}]\frac{e^{-\beta H_{[1,l]}}}{\Tr[e^{-\beta H_{[1,l]}}]}]\right|\\
&\quad+\left|\Tr[\tr_1[E_{l/2}]\left(\frac{e^{-\beta H_{[1,m]}}}{\Tr[e^{-\beta H_{[1,m]}}]}-\frac{e^{-\beta H_{[1,l]}}}{\Tr[e^{-\beta H_{[1,l]}}]}\right)]\right|\\
&\quad+\left|\Tr[\tr_1[E_{m+1}-E_{l/2}]\frac{e^{-\beta H_{[1,m]}}}{\Tr[e^{-\beta H_{[1,m]}}]}]\right|\\
&\le \delta(l/2+1)+\cG \mathcal Ke^{-\alpha l}+\delta(l/2+1)\\
&\le A e^{-l/\xi},
\end{aligned}
\end{equation}
where we used Lemma~\ref{lem:arakiExpansionals} and Lemma~\ref{lem:localIndistinguishability}. The last line follows from upper bounding the superexponential decay by an exponentially decaying function and combining all constants.
Note that
\[
f(h)=-\lim_{m\to\infty} \frac1{\beta m} \log(Z_m)=-\lim_{m\to\infty}\frac1{\beta m}\sum_{i=1}^m\log(\frac{Z_i}{Z_{i-1}}),
\]
where $Z_0=1$, which implies
\[
f(h)=-\frac1\beta\log(\lim_{m\to\infty}\frac{Z_{m+1}}{Z_m})
\]
if the limit exists and is positive.
The existence of the limit follows from its Cauchy convergence that holds due to Equation~\eqref{eq:cauchyCrit} and its positivity  follows from \cite[Lemma 3.5]{lenci2005}
\[
\left|\log(\frac{Z_{m+1}}{Z_m})\right|\le\beta\|h\|.
\]
 Taking the limit of \ref{eq:cauchyCrit} concludes the proof.
\end{proof}
\section{MPO algorithm}
With this expression at hand we can proceed to the proof of our main result.
Deviating from the approach in \cite{kuwahara2018} where the analogue expression from the previous lemma is simply evaluated using exact diagonalization, we use the following algorithm from \cite{kuwahara2021} to construct a matrix product operator approximation of the unnormalized Gibbs states and compute their traces.
\begin{theorem}[{\cite[Theorem 3]{kuwahara2021}}]\label{thm:MPO}
There exists an algorithm taking as input the description of an input hermitian matrix $h\in\CC^{d^2\times d^2}$ with $\|h\|\le1$ and an inverse temperature $\beta$ and outputting a matrix product operator $M$ on $n$ sites such that
\[
\left\|M-e^{-\beta H_{[1,n]}}\right\|_1\le\eps\left\|e^{-\beta H_{[1,n]}}\right\|_1.
\]
The bond dimension is given by
\[
D\le \exp(C\sqrt{\beta\log(n/\eps)}\log(\beta\log(n/\eps))),
\]
and the algorithm runs in time
\[
t\le n\beta\exp(C\sqrt{\beta\log(n/\eps)}\log(\beta\log(n/\eps))),
\]
where $C$ is a numerical constant.
\end{theorem}
\begin{remark}
We note that the cited work contains a case distinction for small and large $\eps$ compared to $\beta$.
Since we always consider $\beta$ as a parameter and only write the dependence on $\eps$, we choose the case of small $\eps$.
This threshold can be absorbed into the constant $B$ in Theorem~\ref{thm:main}.

The Theorem comes with an explicit temperature dependence, which we included in the above formulation.
This dependence will not carry over to our result since the bottleneck is given by the approximation step using Araki expansionals.
\end{remark}
Combining Lemma~\ref{lem:partFuncRatio} with this MPO algorithm we prove our main result.
\begin{proof}[Proof of Theorem~\ref{thm:main}]
The algorithm is given in Algorithm \ref{alg:freeEnergy}.
\begin{algorithm}\label{alg:freeEnergy}
\SetKwInput{KwInput}{Input}
\SetKwInput{KwRequire}{Require}
\SetKwInput{KwOutput}{Output}
\SetKwInput{KwData}{Parameter}
\KwData{inverse temperature $\beta$}
\KwRequire{subroutine ${MPO}_\textrm{Gibbs}$ constructing an MPO approximation of the unnormalized thermal state from Theorem~\ref{thm:MPO}, polynomial time routine to compute the trace of a given MPO $\Tr_{MPO}$ (standard contraction), explicit constants $D,E$ (see below)}
\KwInput{local dimension $d$, Hamiltonian term $h\in\mathbb C^{d^2 \times d^2}$ such that $\|h\|\le1$, error $\eps$}
\KwOutput{$\tilde{f}$ approximation to the free energy $f(h)$}
$l\gets C\log(1/\eps)$\;
$\tilde Z_l\gets \Tr_{MPO}[{MPO}_\textrm{Gibbs}(l,h,\eps\beta\exp(-\beta(2l+1))/12)]$\;
$\tilde Z_{l+1}\gets \Tr_{MPO}[{MPO}_\textrm{Gibbs}(l+1,h,\eps\beta \exp(-\beta(2l+1))/8)]$\;
$\tilde{f}_{\beta} \gets -(1/\beta)\log(\tilde Z_{l+1}/\tilde Z_l)$\;
\caption{Algorithm for computing the free energy density.}
\end{algorithm}

Since all expressions here are ratios of traces, we can switch to normalized traces $\tr[\cdot]=\Tr[\cdot]/d^l$.
Recall, that we denote the partition functions for the system of size $l$ by  $Z_l$.
We further denote by $\tilde Z_l$, $\tilde Z_{l+1}$ the corresponding approximation from the MPO algorithm, i.e., we run the the algorithm from Theorem~\ref{thm:MPO} at an error level $\eps_{MPO}$ to be chosen later and contract the resulting tensor network to obtain a trace estimate. The latter step can be done in time polynomial in the bond dimension, physical dimension and $l$.

Our task is to provide explicit bounds on the parameters $l$ and $\eps_{MPO}$ to obtain the desired accuracy.

The approximation error can be split up into a term caused by the truncation in $l$ and one term due to the approximation errors from the MPO
\begin{equation}\label{eq:errorSplit}
\left|f(h)-\tilde f\right|\le \left|f(h)-\left(-\frac{1}{\beta}\right)\log(\frac{Z_{l+1}}{Z_l})\right|+\left|-\frac{1}{\beta}\log(\frac{Z_{l+1}}{Z_l})-\left(-\frac{1}{\beta}\right)\log(\frac{\tilde Z_{l+1}}{\tilde Z_l})\right|
\end{equation}

In order to control the error propagation for the logarithm and fractions, we need analytic bounds on the partition functions and their ratios.
The absolute values of the traces can be bounded by
\[
e^{-\beta l}\le Z_l\le e^{\beta l}.
\]
Furthermore, from \cite[Lemma 3.6]{lenci2005} we have
\[
\left|\log(\frac{Z_{l+1}}{Z_l})\right|\le\beta.
\]

For the first term in Equation~\eqref{eq:errorSplit} by the mean value theorem we estimate
\[
\left|-\frac1\beta\log(e^{-\beta f})-\left(-\frac{1}{\beta}\right)\log(\frac{Z_{l+1}}{Z_l})\right|\le
\frac1\beta e^{\beta } \left|\frac{Z_{l+1}}{Z_l}-e^{-\beta f}\right|
\]
using the derivative of the logarithm and a lower bound of $e^{-\beta}$ on its argument for all points in the intervall spanned by $e^{-\beta f}$ and $Z_{l+1}/Z_l$.
Invoking Lemma~\ref{lem:partFuncRatio}, we can bound the first term by $\eps/2$ by choosing
\[
l\ge \log(\frac{2e^{\beta}A}{\beta\eps})\xi.
\]

For the second term in Equation~\eqref{eq:errorSplit} we proceed similarly.
To lower bound the derivative of the logarithm we require the error bound
\begin{equation}\label{eq:weakRatioBound}
\left|\frac{Z_{l+1}}{Z_l}-\frac{\tilde Z_{l+1}}{\tilde Z_l}\right|\le \frac12e^{-\beta}
\end{equation}
which will follow from the input accuracy to the MPO subroutine.
Analogously to before, this upper bounds the derivative of the logarithm by $2e^{\beta}$ in the interval with end points $Z_{l+1}/Z_l$, $\tilde Z_{l+1}/\tilde Z_l$.
Under this condition a bound on the second term in Equation~\eqref{eq:errorSplit} reads
\[
\left|-\frac{1}{\beta}\log(\frac{Z_{l+1}}{Z_l})-\left(-\frac{1}{\beta}\right)\log(\frac{\tilde Z_{l+1}}{\tilde Z_l})\right|\le\frac1\beta2e^{\beta}\left|\frac{Z_{l+1}}{Z_l}-\frac{\tilde Z_{l+1}}{\tilde Z_l}\right|\le\frac\eps2,
\]
i.e., we need an error of
\begin{equation}\label{eq:strongRatioBound}
\left|\frac{Z_{l+1}}{Z_l}-\frac{\tilde Z_{l+1}}{\tilde Z_l}\right|\le\frac{\eps\beta e^{-\beta}}4
\end{equation}
for the approximate ratio of partition functions. Note that this is a stronger requirement than Equation~\eqref{eq:weakRatioBound} if $\eps<\beta/2$ so, since we consider only the asymptotics of small $\eps$ for fixed $\beta$, we consider the latter condition included.

Using the inequality
\[
\left|\frac{Z_{l+1}}{Z_l}-\frac{\tilde Z_{l+1}}{\tilde Z_l}\right|\le\frac1{Z_l}\left|Z_{l+1}-\tilde Z_{l+1}\right|+\frac{\tilde Z_{l+1}}{Z_l\tilde Z_l}\left|\tilde Z_l- Z_l\right|,
\]
and recalling
\begin{align*}
Z_l^{-1}&\le e^{\beta l}\\
\frac{\tilde Z_{l+1}}{\tilde Z_l}&\le \frac{Z_{l+1}}{Z_l}-\left|\frac{Z_{l+1}}{Z_l}-\frac{\tilde Z_{l+1}}{\tilde Z_l}\right|\le\frac32 e^{\beta}
\end{align*}
we see that by running the algorithm in Theorem~\ref{thm:MPO} at error levels
\begin{align*}
\left|Z_l-\tilde Z_l\right|&\le \eps_{MPO,l}e^{\beta(l-1)}:=\frac{\eps\beta}{12} e^{-\beta(l+2)}\\
\left|Z_{l+1}-\tilde Z_{l+1}\right|&\le\eps_{MPO,l+1}e^{\beta l}:= \frac{\eps\beta}8 e^{-\beta(l+1)}
\end{align*}
equation~\eqref{eq:strongRatioBound} follows. This proves correctness of the algorithm.
Note that that the inverse input errors for the MPO algorithm are polynomial in $\eps^{-1}$ as $l$ has only logarithmic dependence.

Let us now turn our attention to the runtime of the algorithm.
The relevant part here is the runtime of the algorithm constructing the MPO descriptions of the truncated thermal states and its resulting bond dimension, which polynomially enters the runtime of the contraction.
The remaining steps of the algorithm are evaluations of elementary functions and their contribution to the runtime can be neglected.

From Theorem~\ref{thm:MPO} runtime and bond dimension are given by
\[
t\le l\beta\exp(C \sqrt{\beta\log(l/\eps_{MPO})}\log(\beta\log(l/\eps_{MPO})))
\]
and
\[
D=\exp(C \sqrt{\beta\log(l/\eps_{MPO})}\log(\beta\log(l/\eps_{MPO})))
\]
$\eps_{MPO}$ is the input error for the MPO construction algorithm.
The runtime of the contraction is $\poly(l,D, d)$.
The temperature dependence of the runtimes can be simplified by combining all constants into constants $\tilde A,\tilde B$ and we obtain an overall bound on the runtime
\[
t_\textrm{total}\le \exp(\sqrt{\tilde A\log(\frac{\tilde B}\eps)}\log(\tilde A\log(\frac{\tilde B}\eps))),
\]
which closes the proof.

\end{proof}
\paragraph{Acknowledgements}
I would like to thank {\'{A}}lvaro Alhambra and Hamza Fawzi for helpful discussions. I acknowledge support from the UK Engineering and Physical Sciences Research Council (EPSRC) under grant number EP/W524141/1.

\clearpage
\printbibliography
\appendix
\section{Quantum Belief Propagation}
In this Appendix we give an alternative proof of Lemma~\ref{lem:partFuncRatio} based on \cite{kuwahara2018}.
We deviate in a few technical details, in particular, we do not need a result on the decay of correlations.
While the decay of correlations has been proven for translation-invariant infinite~\cite{araki1969}, and translation-invariant finite~\cite{bluhm2022} Gibbs states in one dimension, the requirement in \cite{kuwahara2018} is for a finite but perturbed Gibbs state.
We believe that this gap could be filled, however, our approach circumvents this issue by making use of a known local indistinguishability result instead.

The significantly simpler proof in the main text is restricted to the translation-invariant one-dimensional setting.
The Quantum Belief Propagation (QBP) used in this chapter is a general perturbation formula that applies to arbitrary quantum systems satisfying Lieb-Robinson bounds.
However, to conclude a result equivalent to Lemma~\ref{lem:partFuncRatio}, we require the additional property of local indistinguishability that is satisfied, i.e., for systems at high temperature, but usually also comes with alternative algorithmic approaches such as cluster expansion \cite{Mann2021}.

This section is structured as follows.
We start by giving the exact QBP formula for a perturbation equivalent to adding sites in the partition function.
The formula gives a perturbation operator that relates a perturbed unnormalized Gibbs state to the original one.
The second step is to prove locality properties of these perturbation operators.
Here we slightly deviate from previous formulations and give an approximation in terms of truncated Hamiltonians rather than partial traces of the perturbation operators.
We believe this formulation could also be convenient for further work using the QBP.
We conclude by combining this formula with the local indistinguishability in Lemma~\ref{lem:localIndistinguishability}.

\subsection{Exact Quantum Belief Propagation}\label{sec:QBP}
The Quantum Belief Propagation is originally due to \cite{hastings2007}.
It gives an analytic expression for a perturbed (unnormalized) Gibbs state.
We will use this result to express the ratio of partition functions with an additional site as of Lemma~\ref{lem:partFuncRatio}.

We make use of the time evolution operator for a given Hamiltonian
\[
\Gamma^t_{H_0}(A)=e^{itH_0}Ae^{-itH_0}
\]
and start by giving the main definitions necessary for the QBP.
\begin{definition}[Exact Quantum Belief Propagation]
The following equations define the Quantum Belief Propagation Operators for any perturbed Hamiltonian $H(s)=H_0+sV$, $0\le s\le1$.
\begin{align*}
\Phi_\beta^{H(s)}(V)&=\int_{-\infty}^\infty f_\beta(t) e^{iH(s)t}Ve^{-iH(s)t}dt\\
\tilde f_\beta(\omega)&=\begin{cases}
\frac{\tanh(\beta\omega/2)}{\beta\omega/2}&\textrm{ for }\omega\neq0\\
1&\textrm{ for }\omega=0
\end{cases}
\end{align*}
The inverse Fourier transform is given by \cite[Appendix B]{anshu2021}
\[
f_\beta(t)=\frac1{2\pi}\int_{-\infty}^\infty e^{i\omega t}\tilde f_\beta(\omega)dt=\frac{2}{\beta\pi}\log\frac{e^{\pi|t|/\beta}+1}{e^{\pi|t|/\beta}-1}.
\]
Now this defines the QBP operators as follows
\begin{align*}
\eta&=\eta(1)\\
\eta(s)&=\mathcal{T}\exp(-\frac\beta2\int_0^s \Phi^{H(s)}_\beta(V)ds)\\
&=\id+\sum_{k\ge1}\left(-\frac{\beta}2\right)^k\int_0^s\int_0^{s_1}\cdots\int_0^{s_{k-1}}\Phi^{H(s_1)}_\beta(V)\ldots\Phi^{H(s_k)}_\beta(V)ds_1\ldots ds_k.
\end{align*}
\end{definition}
Note that $\|f_\beta\|_{L^1}=1$ and therefore a straightforward application of the triangle inequality yields
\[
\left\|\Phi^{H(s)}_\beta(V)\right\|\le\|V\|.
\]

In the following, we will specialize our statements to the choice of Hamiltonian and perturbation relevant for us.
To reduce the amount indices we fix a length scale $L$ and introduce
\begin{align*}
H(s)&=H_{[1,L]}+sh_{L,L+1}\\
H_l(s)&=H_{[L-l,L]}+sh_{L,L+1}.
\end{align*}
By fixing the energy units we keep the assumption $\|h\|\le1$. For the sake of brevity we refrain from giving the explicit dependence of constants on $\beta$ in the following.
We start from the following formulation of Quantum Belief Propagation for $H_0=H(0)$ and $V=h_{L,L+1}$.
Note that in the following, all traces are understood to act on the system $[1,L+1]$, which is compensated for by a factor $d$.
\begin{lemma}[{\cite[Proposition 6]{capel2023}}]\label{lem:QBP_F}
We have
\[
\frac{Z_{L+1}}{Z_L}=d\frac{\Tr[\exp(-\beta H(1))]}{\Tr[\exp(-\beta H(0))]}=d\Tr\left[\eta\frac{\exp(-\beta H(0))}{\Tr[\exp(-\beta H(0))]}\eta^\dagger\right]
\]
\end{lemma}
Our aim in the following subsections is to quantify the locality of the QBP operators in order to approximate this expression by an expectation value of a local observable in the Gibbs state of a smaller system.
\subsection{Locality of the QBP Operators}
The first step is to approximate the operators $\Phi_\beta^{H(s)}$ using time evolution operators of truncated Hamiltonians.
Although the proof is closely related to the known approach \cite{capel2023}, the following Lemma provides a slightly different formulation.
Rather than projecting $\eta$ to a local operator, we redefine it using an approximation of the time evolution.
\begin{lemma}\label{lem:QBPLocality}
There are some constants $A,B$ depending only on $\beta$ such that for $0<\eps<1$ and for
\[
l\ge A \log(\frac1\eps)+B
\]
we have
\[
\left\| \Phi_\beta^{H(s)}(h_{L,L+1})-\Phi_\beta^{H_l(s)}(h_{L,L+1}) \right\|\le \eps.
\]
\end{lemma}
\begin{proof}
The defintion of $\Phi^{H(s)}_\beta(h_{L,L+1})$ contains time-evolved interactions for all times with a time-depdent weight $f_\beta(t)$.
We split the definition into short time scales, for which we use a Lieb-Robinson bound (see Appendix~\ref{sec:liebRob}), and long time-scales, whose contribution can be bounded using the decay of $f_\beta(t)$.
\begin{align*}
\left|\Phi_\beta^{H(s)}(h_{L,L+1})-\Phi_\beta^{H_l(s)}(h_{L,L+1})\right|&\le\sup_{|t|\le a}\left\|\Gamma^t_{H(s)}(h_{L,L+1})-\Gamma^t_{H_l(s)}(h_{L,L+1})\right\|\\
&\phantom=+4\left\|h_{L,L+1}\right\|\int_a^\infty f_\beta(t)dt\\
\end{align*}
The integral can be upper bounded by (see \cite[Appendix B]{alhambra2023})
\[
\int_a^\infty f_\beta(t)dt\le\frac{4}{\pi^2(e^{\pi a/\beta}-1)}\le\frac\eps8
\]
if we set 
\[
a=\frac\beta\pi\log(\frac{8}{\pi^2\eps}+1).
\]
Now, by Theorem \ref{thm:liebRob} setting
\begin{align*}
l&\ge\frac{E\beta}{\pi}e^{2D}\log(\frac8{\pi^2\eps}+1)\\
l&\ge\frac{C\beta}{D\pi}\log(\frac8{\pi^2\eps}+1)+\frac1D\log(\frac4\eps)\\
\end{align*}
and choosing $D=1$, we have that the norm bound on the first term is also $\le\eps/2$.
Combining the constants in the conditions on $l$ closes the proof.
\end{proof}
Replacing the $\Phi$ by these local versions in the definition of QBP operators immediately yields local approximations for those.
\begin{lemma}\label{lem:etaLocality}
For $l\ge A\log(1/\eps)+B$, where $A,B$ are $\beta$-dependent constants, we have
\begin{align*}
\|\eta-\eta_l\|&\le\eps\\
\|\eta_l\|,\|\eta\|&\le\exp(\beta/2).
\end{align*}
Here, $\eta_l$ is an operator with support on sites $L-l\ldots L+1$:
\[
\eta_l=\mathcal{T}\exp(-\frac\beta2\int_0^1 \Phi^{H_l(s)}_\beta(h_{L,L+1})ds)
\]
\end{lemma}
\begin{proof}
Recall that
\[
\left\|\Phi^{H(s)}_\beta(h_{L,L+1})\right\|,\left\|\Phi^{H_l(s)}_\beta(h_{L,L+1})\right\|\le\|h_{L,L+1}\|\le1.
\]
By using the triangle inequality we find
\begin{align*}
\|\eta-\eta_l\|&\le\sum_{k\ge1}\left(\frac\beta2\right)^k\int_0^1\int_0^{s_1}\cdots\int_0^{s_{k-1}}\Big\|\Phi^{H(s_1)}_\beta(h_{L,L+1})\ldots\Phi^{H(s_k)}_\beta(h_{L,L+1})\\
&\qquad-\Phi^{H_l(s_1)}_\beta(h_{L,L+1})\ldots\Phi^{H_l(s_k)}_\beta(h_{L,L+1})\Big\|ds_1\ldots ds_k\\
&\le\sum_{k\ge1}\left(\frac\beta2\right)^k\int_0^1\int_0^{s_1}\cdots\int_0^{s_{k-1}}\|h_{L,L+1}\|^{k-1}\\
&\qquad\sum_{m=1}^k \|\Phi^{H(s_m)}_\beta(h_{L,L+1})-\Phi^{H_l(s_m)}_\beta(h_{L,L+1})\|ds_1\ldots ds_k\\
&\le \eps_{QBP}\sum_{k\ge1} \frac{k(\beta/2)^k}{k!}\\
&=\eps_{QBP} \frac\beta2e^{\beta/2}\\
&\le\eps,
\end{align*}
where we assumed
\[
l\ge A\log(\frac{\beta e^{\beta/2}}{2\eps})+B
\]
for $A,B$ such that $\eps_{QBP}\le \eps \frac2\beta e^{-\beta/2}$ from Lemma~\ref{lem:QBPLocality}.
Redefining the constants we conclude the proof.
\end{proof}
\subsection{Local Indisinguishability and Alternative Proof of Lemma~\ref{lem:partFuncRatio}}

In the previous sections we showed that the ratio of partition functions can be written in terms of QBP operators and that those have approximations by local operators.
Connecting this with the local indistinguishability that we already saw in Lemma~\ref{lem:localIndistinguishability}, we can reprove Lemma~\ref{lem:partFuncRatio}, which we restate here.
\begin{lemma}[Restatement of Lemma~\ref{lem:partFuncRatio}]
For some $\beta$-dependent constants $A,\xi$, we have
\[
\left|e^{-\beta f(h)}-\frac{\Tr[\exp(-\beta H_{[1,l+1]})]}{\Tr[\exp(-\beta H_{[1,l]})]}\right|\le\eps
\]
if $l\ge \xi\log(A/\eps)$
\end{lemma}
\begin{proof}
We will start from the ratio of partition functions at finite $L\ge l$.
The following bounds will hold uniformly in $L$ and thereby give the desired result in the limit.
Note that in the following the traces in each fraction are on equally large subsystems which causes the factors of $d$.
By writing $l=k_1+k_2$ as a sum of the decay lengths for the QBP and the local indistinguishability, we estimate
\begin{align}
&\left|\frac{Z_{L+1}}{Z_L}-\frac{Z_{l+1}}{Z_l}\right|\\
&\qquad=\left|d\frac{\Tr[\exp(-\beta H(1))]}{\Tr[\exp(-\beta H(0))]}-d\Tr[\eta_{k_1+k_2}\frac{\exp(-\beta H_{k_1+k_2})}{\Tr[\exp(-\beta H_{k_1+k_2})]}\eta_{k_1+k_2}^\dagger]\right|\\
&\qquad\le d\left|\Tr[\eta\frac{\exp(-\beta H_{[1,L]})}{\Tr[\exp(-\beta H_{[1,L]})]}\eta^\dagger]-\Tr[\eta_{k_1}\frac{\exp(-\beta H_{[1,L]})}{\Tr[\exp(-\beta H_{[1,L]})]}\eta_{k_1}^\dagger]\right|\label{eq:firstQBP}\\
&\qquad\phantom{\le}+d\left|\Tr[\eta_{k_1}\frac{\exp(-\beta H_{[1,L]})}{\Tr[\exp(-\beta H_{[1,L]})]}\eta_{k_1}^\dagger]-\Tr[\eta_{k_1}\frac{\exp(-\beta H_{k_1+k_2})}{\Tr[\exp(-\beta H_{k_1+k_2})]}\eta_{k_1}^\dagger]\right|\label{eq:secondQBP}\\
&\qquad\phantom{\le}+d\left|\Tr[\eta_{k_1}\frac{\exp(-\beta H_{k_1+k_2})}{\Tr[\exp(-\beta H_{k_1+k_2})]}\eta_{k_1}^\dagger]-\Tr[\eta_{k_1+k_2}\frac{\exp(-\beta H_{k_1+k_2})}{\Tr[\exp(-\beta H_{k_1+k_2})]}\eta_{k_1+k_2}^\dagger]\right|\label{eq:thirdQBP}
\end{align}
Despite being defined as the \emph{approximate} QBP operator for $H_{[1,L]}$, note that in the first line we used the fact that by its definition $\eta_l$ is the \emph{exact} QBP for the Hamiltonian $H_l$.
This feature is due to our definition of $\eta_l$ in terms of truncated time-evolution operators.

For the term~\eqref{eq:firstQBP}, note that it is an expectation value of an observable that obeys the norm bound
\[
\left\|\Tr_{L+1}\left[\eta^\dagger\eta-\eta_{k_1}^\dagger\eta_{k_1}\right]\right\|\le d\left\|\eta^\dagger\eta-\eta_{k_1}^\dagger\eta_{k_1}\right\|\le d(\left\|\eta\right\|+\left\|\eta_{k_1}\right\|)\|\eta-\eta_{k_1}\|\le\frac\eps3
\]
by choosing $k_1$ according to Lemma~\ref{lem:etaLocality}.
Analogously by applying the locality result twice we bound the term~\eqref{eq:thirdQBP}.

The term~\eqref{eq:secondQBP} can be bounded using local indistinguishability as of Lemma~\ref{lem:localIndistinguishability} applied to the local observable
\[
\Tr_{L+1}\left[\eta_{k_1}^\dagger\eta_{k_1}\right],
\]
which obeys the bound from Lemma~\ref{lem:etaLocality} and for the decay length $k_1+k_2-1$.
Hence, adding a condition to the minimum size of $k_1$ (or $k_2$), we can bound this term by $\eps/3$ as well.

Combining all the conditions on $k_1$ and $k_2$, which gives a sufficient condition for $l$ and taking the limit $L\to\infty$, we obtain the result.
\end{proof}

\section{Lieb-Robinson Bounds for Truncated Time Evolution}\label{sec:liebRob}
We restate Theorem 2.4 from \cite{perezgarcia2020}.
We modify the formulation in that we assume 2-local interactions instead of only exponentially decaying ones.
For finite interactions the result becomes valid for all temperatures.
The result does \emph{not} require translation-invariance, but we only state it for our choice of perturbed Hamiltonian.

We recall the time evolution operator
\[
\Gamma_H^t(A)=e^{itH}Ae^{-itH}
\]
\begin{theorem}[{\cite[Theorem 2.4]{perezgarcia2020}}]\label{thm:liebRob}
For a an operator $A$ with support on $L,L+1$ and $\beta$-dependent constants $C, E$,
\begin{align*}
\left\|\Gamma_{H(s)}^t(A)-\Gamma_{H_l(s)}^t(A)\right\|&\le\|A\| 2 e^{8|t|\beta}\sum_{k\ge l+1}\Omega_k^*(4|t|)\\
&\le 2 \|A\| e^{C|t|-Dl}.
\end{align*}
for any choice of $D>0$ as long as $l\ge E|t|e^{2D}$.
\end{theorem}
\begin{proof}
We recall some definitions from \cite{perezgarcia2020} specialising to the case of a 2-local 1D Hamiltonian
\[
\Omega_0=\Omega_1=\Omega_2=2\beta\qquad\qquad0=\Omega_3=\ldots
\]
\begin{align*}
\Omega_k^*(x)&=\sum_{n=1}^\infty\left(\sum_{\substack{\alpha\in\mathbb{N}^n\\|\alpha|=k}}\prod_{j=1}^n\Omega_{\alpha_j}\right)\frac{x^n}{n!}\\
&\le\sum_{n\ge\lceil k/2\rceil}\frac{(6x\beta)^n}{n!}
\end{align*}
We make repeated use of the tail bound for the exponential series
\[
\sum_{k\ge l}\frac{x^k}{k!}\le \frac{e^x x^l}{l!}
\]
We start from \cite[Theorem 2.4]{perezgarcia2020}
\begin{align*}
\|\Gamma_{H(s)}^t(A)-\Gamma_{H_l(s)}^t(A)\|&\le\|A\| |\supp(A)| e^{8|t|\beta}\sum_{k\ge l+1}\Omega_k^*(4|t|)\\
&\le2\|A\|e^{8|t|\beta}\sum_{k\ge l+1}e^{24|t|\beta}\frac{(24|t|\beta)^{\lceil k/2\rceil}}{\lceil k/2\rceil!}\\
&\le2\|A\|e^{32|t|\beta}e^{24|t|\beta}\frac{(24|t|\beta)^{\lceil l/2\rceil}}{\lceil l/2\rceil!}\\
&\le2\|A\|e^{56|t|\beta}e^{\lceil l/2\rceil\log(24|t|\beta)-\lceil l/2\rceil\log(\lceil l/2\rceil)}\\
\end{align*}
Choosing $l\ge48|t|\beta e^{2D}=E|t|e^{2D}$ we find
\[
\|\Gamma_{H(s)}^t(A)-\Gamma_{H_l(s)}^t(A)\|\le 2\|A\|e^{C|t|-Dl}
\]
where $C=56\beta$.
\end{proof}

\end{document}